\documentclass[11pt]{article}
 
\newtheorem{theorem}{Theorem}[section]
\newtheorem{lemma}[theorem]{Lemma}
\newtheorem{corollary}[theorem]{\indent Corollary}

\newcommand{\qed}{\nopagebreak \hfill $\Box$}
\newenvironment{proof}{\par \noindent {\em Proof:}}{\qed \par}

  {\begin{center}%
   \begin{minipage}[b]{0.8\textwidth}
   \begin{list}{}%
    {\setlength{\listparindent}{-2.5mm}%
     \setlength{\itemindent}{\listparindent}}%
     \setlength{\labelsep}{0in}%
     \setlength{\parskip}{0ex}%
     \setlength{\parsep}{\parskip}%
     \setlength{\topsep}{0in}%
     \setlength{\partopsep}{0in}%
   \item[]%
  }%
  {\end{list}%
   \end{minipage}%
   \end{center}%
  }
 
  {\begin{list}{}%
    {\setlength{\leftmargin}{5mm}%
     \setlength{\listparindent}{-0.5\leftmargin}%
     \setlength{\itemindent}{\listparindent}}%
     \setlength{\labelsep}{0in}%
     \setlength{\parskip}{0ex}%
     \setlength{\parsep}{\parskip}%
     \setlength{\topsep}{0in}%
     \setlength{\partopsep}{0in}%
    \item[]%
  }%
  {\end{list}%
  }
 
\newcommand{\ol}{\setlength{\itemsep}{0pt.}\begin{enumerate}}
\newcommand{\eol}{\end{enumerate}\setlength{\itemsep}{-\parsep}}
\usepackage{graphicx}
\usepackage{latexsym}
\usepackage{algorithmic}

\usepackage{psfrag}
\usepackage{color}
\usepackage[section]{placeins}

\title{LP-rounding algorithms for facility-location problems}

\author{
Jaroslaw Byrka\thanks{Institute of Computer Science, University of Wroclaw, Poland. {\tt jby@ii.uni.wroc.pl}.
Work of this author was partially conducted at CWI Amsterdam, TU Eindhoven, EPFL Lausanne, and while visiting
the University of Maryland. Partially supported by FNP HOMING PLUS/2010-1/3 grant and MNiSW grant number N N206 368839,
2010-2013.
}
\and
MohammadReza Ghodsi\thanks{Dept.\ of Computer Science, University of Maryland,
College Park, MD 20742, USA. \tt{ghodsi@cs.umd.edu}}
\and
 Aravind Srinivasan\thanks{Dept.\ of Computer Science and
Institute for Advanced Computer Studies, University of Maryland,
College Park, MD 20742, USA. {\tt srin@cs.umd.edu}.
Supported by NSF Award CNS-0626636.}}

\newcommand{\C}{\mathcal{C}}
\newcommand{\F}{\mathcal{F}}

\begin{document}
\maketitle

\begin{abstract}
We study LP-rounding approximation algorithms for 
metric uncapacitated facility-location problems.
We first give a new analysis for the algorithm of Chudak and Shmoys, which
 differs from the analysis of Byrka and Aardal in that
now we do not need any bound based on the solution to the dual LP program.
Besides obtaining the optimal bifactor approximation as do
Byrka and Aardal, we can now also show that the algorithm with scaling 
parameter equaling $1.58$ is, in fact, an $1.58$-approximation algorithm.
More importantly, we suggest an approach based on additional randomization
and analyses such as ours, which could achieve or approach the conjectured
optimal $1.46\cdots$--approximation for this basic problem. 

Next, using essentially the same techniques,
we obtain improved approximation algorithms in the $2$-stage stochastic 
variant of the problem, where we
must open a subset of facilities having only stochastic
information about the future demand from the clients.
For this problem
we obtain a $2.2975$-approximation algorithm in the standard setting,
and a $2.4957$-approximation in the more restricted, per-scenario setting.

We then study robust fault-tolerant facility location, 
introduced by Chechik and Peleg: solutions here are designed to provide low connection cost in 
case of failure of up to $k$ facilities.
Chechik and Peleg gave a
$6.5$-approximation algorithm for $k=1$ and a ($7.5k + 1.5$)-approximation algorithm
for general $k$.
We improve this to an LP-rounding $(k+5+4/k)$-approximation algorithm.
We also observe that in case of oblivious failures the expected approximation ratio
can be reduced to $k + 1.5$, and that the integrality gap of the natural
LP-relaxation of the problem is at least $k + 1$. 

\end{abstract}

\setcounter{page}{0}
\thispagestyle{empty}

\newpage

\section{Introduction} 
In facility location problems, we seek a subset of given locations
where to build facilities, in order to service a given set of clients.
The goal is to minimize the total cost of constructing facilities
and the clients' service cost, which in the metric setting is a function of distances
between clients and the facilities that they are assigned to.
In this paper we will only consider uncapacitated problems,
where there is no restriction on the number of clients connected
to a single facility. (We sometimes use the terminology of 
\emph{opening} a subset of the existing facilities, rather then 
\emph{constructing} them.) We present improved approximation
algorithms for a variety of such problems using LP-rounding; we also
sketch a possible approach to achieving or approaching the optimal
approximation for the basic and perhaps most-studied variant, the
Uncapacitated Facility Location problem (UFL). 

The UFL is defined as follows.
Given a set $\mathcal{F}$ of facilities and a set $\mathcal{C}$ 
of clients, we aim to open a subset of facilities and connect every
client to an open facility. The cost of opening facility $i$ is $f_i$, 
and the cost of connecting client $j$ to facility $i$ is $c_{ij}$;
the connection costs are assumed to define a symmetric metric.
The goal is to choose the facilities and connections
 that minimize the sum of facility-opening costs and client-connection costs. 

The UFL problem is NP-hard and hard to approximate better than the
positive solution $s_0 \sim 1.46$ to the equation $s = 1 + 2 e^{-s}$~\cite{DBLP:journals/jal/GuhaK99}, where $e$ denotes the base of the natural logarithm. 
Jain et al.~\cite{DBLP:journals/jacm/JainMMSV03} generalized this result to the following bifactor lower bound.
Let $F_{OPT}$ and $C_{OPT}$ denote the facility-opening cost and client-connection cost of an optimal solution. 
They proved that it is unlikely, for any $\lambda \geq 1$, that there exists 
a polynomial-time algorithm that finds a solution with facility-opening cost at most
$\lambda \cdot F_{OPT}$ { and} connection cost at most $(1 + 2e^{-\lambda})\cdot C_{OPT}$. 
On the positive side, Shmoys, Tardos and Aardal~\cite{DBLP:conf/stoc/ShmoysTA97} provided
a constant factor LP-rounding approximation algorithm that exploits
the assumption on the connection costs being metric.
A series of algorithms improving the approximation ratio then
followed, borrowing and contributing to essentially all known major
approaches in approximation algorithms.  
The currently best-known approximation ratio is reached by the 
$1.5$-approximation algorithm of
Byrka and Aardal~\cite{byrka:2212}: this is obtained by a combination of a greedy algorithm
analyzed by a dual fitting technique~\cite{DBLP:journals/siamcomp/MahdianYZ06} and a novel analysis of the 
LP-rounding algorithm of Chudak and Shmoys~\cite{DBLP:journals/siamcomp/ChudakS03}.

The first result of this paper is to simplify the analysis of Byrka and Aardal.
In our case, the expected connection cost of a single client 
gets bounded with respect to the fractional connection cost of the client,
rather than w.r.t.\ a combination of its fractional connection and a dual budget
as in~\cite{DBLP:journals/siamcomp/ChudakS03,byrka:2212}. The main result of~\cite{byrka:2212} 
remains unaffected: we still
obtain that the expected cost of the solution is at most $\gamma \cdot F_{OPT} + (1 + 2e^{-\gamma})\cdot C_{OPT}$ 
if the scaling parameter is $\gamma \geq 1.678$. However, our analysis is
purely primal-based, and thus works if we have an approximately-optimal
LP solution as well. (That is, our approximation bounds will be scaled by
$(1 + \epsilon)$ if we have an $(1 + \epsilon)$-approximate solution to the
LP -- obtained, for instance, by some fast algorithm.) 
Furthermore, for smaller values of the parameter $\gamma$, we obtain bounds that are stronger than in~\cite{byrka:2212}.
In particular, for $\gamma=1.575$ we obtain 1.575-approximate solutions which was not known
before. Interestingly, the same ratio was previously obtained by yet another analysis,
namely an analysis of Sviridenko~\cite{DBLP:conf/ipco/Sviridenko02}, who considered essentially the same algorithm, but
with the scaling parameter $\gamma$ drawn randomly from a certain nontrivial distribution. Perhaps most importantly, we suggest a new type of approach 
for the UFL based on our analysis, which appears promising in terms of
approaching the optimal approximation of $s_0 \sim 1.46\cdots$. 
 
Next we consider the setting of uncertain, stochastic demand modeled as a 2-stage stochastic
optimization problem. In the first stage, given stochastic information about the set of clients
that needs to be served we decide to open a subset of facilities. Next, in the second stage,
the actual set of clients is revealed to us and we can open additional facilities.
Finally we connect each client to a facility opened in any of the stages.
The essence of the problem is that facility-opening costs change over time,
i.e., it is cheaper to open a facility earlier. We make the standard assumption
that the stochastic demand is presented to us in the form of a polynomial number
of possible scenarios, each scenario to be realized with a certain probability.
The goal is to minimize the total expected cost. Certain algorithms deliver slightly stronger,
per-scenario bounds, i.e., the cost in each scenario is compared to the fractional cost in this scenario.
The 2-stage stochastic facility location was introduced by Swamy and Shmoys~\cite{DBLP:conf/focs/SwamyS05}.
The approximation ratio was then improved by Srinivasan~\cite{DBLP:conf/soda/Srinivasan07}, who obtained
a $2.369$-approximation in the general expectation setting and a $3.095$-approximation
in the per-scenario model. We use the techniques we develop for UFL
to improve these ratios to $2.2975$ 
and $2.4957$ respectively.

Finaly, we consider the Robust Fault-Tolerant Facility Location (RFTFL) problem
introduced recently by Chechik and Peleg~\cite{DBLP:conf/stacs/ChechikP10}, 
and apply some insights from stochastic facility location. 
In RFTFL, one has to choose a set of facilities that are in a sense robust: 
i.e., in case of
failure of up to $k$ of the opened facilities, where $k$ is viewed as
a constant, the cost of connecting clients
to the facilities that did not fail should be small. More precisely,
we bound the total facility-opening cost plus a worst case client-connection cost.
We start by observing that this problem can be modeled by an IP 
similar to the one used for the 2-stage stochastic problem. Now we say that facilities
are opened only in the first stage and there are $|\mathcal{F}| \choose k$ scenarios,
each of them excluding the use of a certain subset of $k$ facilities.
We present an LP-rounding $(k+5+4/k)$-approximation algorithm,
which improves (for $k > 1$) upon the bound of $7.5k + 1.5$ from~\cite{DBLP:conf/stacs/ChechikP10}.
We also show that if the scenario is chosen by an oblivious 
adversary, the bound can be improved to $k + 1.5$.
Finally, we show a natural limit of this LP-rounding method
by constructing simple instances for which the integrality gap of such an
LP-relaxation of the problem is at least $k+1$.

\section{Uncapacitated Facility Location problem}


We start with a sketch of the algorithm, and then discuss the crucial 
steps in more detail. Following this, we sketch our ideas for approaching
the optimal $s_0 \sim 1.46\cdots$-approximation. 

By $CS(\gamma)$, we denote the algorithm of 
Chudak and Shmoys~\cite{DBLP:journals/siamcomp/ChudakS03}
with the scaling parameter equaling $\gamma$. A sketch of the 
$CS(\gamma)$ algorithm is as follows: 

\begin{enumerate}
\item
    Solve the standard LP-relaxation (see below) of UFL. 
\item \label{point:mod}
    Modify the fractional solution by: 
\begin{itemize}
 \item scaling up the facility-opening variables by $\gamma$,
 \item modifying the connection variables to completely use the ``closest''
	fractionally open facilities,
 \item splitting facilities, if necessary, such that there is no slack
       between the amount that a client is assigned to a facility, and the
       amount by which this facility is opened.
\end{itemize}
\item \label{point:cluster}  
    Divide clients into clusters based on the current fractional solution.
    In each cluster a specific client is assigned to be a ``cluster center''.
    (\emph{This is a key step.})
    
\item \label{point:open1}
    For every cluster, open one of the ``nearby'' facilities of the cluster center.
\item
    For each facility not considered above,
    open it independently with probability equal to its (scaled) 
fractional opening value. 
\item
    Connect each client to an open facility that is closest to it.
\end{enumerate}

\paragraph{IP formulation and relaxation.}

UFL has a natural formulation as the following integer program.
\begin{eqnarray}
\nonumber
\mbox{min  }&\sum_{i \in \mathcal{F}, j \in \mathcal{C}}{c_{ij}x_{ij}} +
\sum_{i \in \mathcal{F}}{f_i y_i} & ~\mbox{s.t.} \\ \nonumber
&\sum_{i \in \mathcal{F}} x_{ij} = 1 & \mbox{for all }j \in \mathcal{C}, \\ \nonumber
&x_{ij} - y_i \leq 0 & \mbox{for all } i \in \mathcal{F}, j \in \mathcal{C},\\ \label{eq_int}
&x_{ij},y_i \in \{ 0,1 \}       & \mbox{for all } i \in \mathcal{F}, j \in \mathcal{C}\,.
\end{eqnarray}

A linear relaxation of this IP formulation is obtained by
replacing the integrality constraints (\ref{eq_int}) by the
constraint
$x_{ij} \geq 0$, $y_i \geq 0$  for all $i \in \mathcal{F}, j \in \mathcal{C}$.
We use this LP relaxation as a
lower bound for the cost of the optimal integral solution.

\paragraph{Scaling and clustering.}
Inspired by a filtering technique of Lin and Vitter,
the following scaling procedure has been successfully applied
to facility location problems.
Suppose that we have solved the LP
relaxation, and that the optimal~\footnote{Our algorithm is entirely primal and therefore
may start with an arbitrary feasible fractional solution. we only start with an optimal fractional solution
to be sure that its cost is no more then the cost of the optimal integral one.} primal solution is $(x^*,y^*)$.
We will start by modifying 
$(x^*,y^*)$ by scaling the $y$-variables by a constant $\gamma >
1$ to obtain a fractional solution $(x^*,\tilde{y} )$, where $
\tilde{y} = \gamma \cdot y^*$.
Note that by scaling we might set some $\tilde{y}_i > 1$.
In the filtering of Shmoys et al.~such a variable
would instantly be rounded to 1. However, for the compactness
of a later part of our analysis it is useful not to round
these variables, but rather to split facilities. 

Before we discuss
splitting, let us first modify the connection variables.
Suppose that the values of the
$y$-variables are scaled and fixed, but that we now have the freedom to
change the values of the $x$-variables in order to minimize the
connection cost.
For each client $j$ we compute the values of the corresponding
$\tilde{x}$-variables in the following way. We choose an ordering
of facilities with non-decreasing distances to client $j$. We
connect client $j$ to the first facilities in the ordering so that
among the facilities fractionally serving $j$, only the last one
can be opened by more than that it serves
$j$ (i.e., for any facilities $i$ and $i'$ such that $i'$ is
later in the ordering, if $\tilde{x}_{ij} < \tilde{y}_i$ then
$\tilde{x}_{i'j}=0$).
In the next step, we eliminate the occurrences of
situations where $0 < \tilde{x}_{ij} < \tilde{y}_i$. We do so by
creating an equivalent instance of the UFL problem, where facility
$i$ is split into two identical facilities $i'$ and $i''$. In the
new setting, the opening of facility $i'$ is $\tilde{x}_{ij}$ and
the opening of facility $i''$ is $\tilde{y}_i - \tilde{x}_{ij}$.
The values of the $\tilde{x}$-variables are updated accordingly.
By repeatedly applying this procedure we obtain a so-called
\emph{complete} solution $(\overline{x}, \overline{y})$, i.e., a
solution in which no pair $i \in \mathcal{F}, j \in \mathcal{C}$
exists such that $0 < \overline{x}_{ij} < \overline{y}_i$
(see~\cite{DBLP:conf/ipco/Sviridenko02}[Lemma 1] for a more detailed argument).


Based on the complete fractional solution $(\overline{x}, \overline{y})$, some of the 
facilities in our instance of UFL are grouped into clusters.
(It is sometimes intuitive to view clusters as sets of clients instead.)
Each cluster of facilities is created by picking a client as {\em cluster center} and
creating a cluster from the facilities serving it in 
$(\overline{x}, \overline{y})$. 
We require that no facility belong to more than one cluster. Therefore when a cluster center is picked, 
any client that shares a facility with the cluster center can no longer be picked as the center of a new cluster.
The algorithm of \cite{DBLP:journals/siamcomp/ChudakS03} uses the following procedure to obtain the
clustering: while not all the clients are processed, greedily choose
(in the manner described next) a new cluster center $j$, and build 
a cluster from $j$
and facilities serving $j$ in $(\overline{x}, \overline{y})$. 
Remove $j$ and any client that shares a facility with $j$.
The greedy choice of the next cluster center depends on the distences between the clients 
and facilities serving them. For each client $j$, compute: 
(i) $d_j^{(c)}$, the average distance from $j$ to facilities serving it 
in $(\overline{x}, \overline{y})$, and (ii) $d_j^{(max)}$,  
the maximum distance to any facility serving it in 
$(\overline{x}, \overline{y})$.
(For more formal definitions see the analysis of the algorithm.)
The next cluster center is the remaining client with smallest $d_j^{(c)} + d_j^{(max)}$.

To obtain the integral solution, we round the fractional variables as follows.
For each cluster the complete fractional opening variables $\overline{y}_i$ sum to 1. We 
open exactly one facility within each cluster, with probabilities equal to the $\overline{y}_i$s.
Any facility that is not in a cluster is opened independently 
with probability $\overline{y}_i$.
Each client is connected to the closest open facility.

\subsection{Analysis}


We will use the average distances between single clients and groups of facilities
defined as follows.
%
  For any client $j \in \mathcal{C}$, and for any subset of facilities $\mathcal{F'} \subset \mathcal{F}$
  such that $\sum_{i \in \mathcal{F'}}{\overline{y}_i} > 0$, let
$
    d(j,\mathcal{F'}) = \frac{\sum_{i \in \mathcal{F'}}{c_{ij} \cdot \overline{y}_i}}
                           {\sum_{i \in \mathcal{F'}}{\overline{y}_i}}.
$ 
%
%
We will call the set of facilities $i \in \F$ such that $\overline{x}_{ij} > 0$ the set of \emph{close
facilities} of client $j$ and we denote it by $C_j$. By analogy, we will call the set of facilities $i \in \F$
such that $x^*_{ij} > 0$ and $\overline{x}_{ij} = 0$ the set of \emph{distant facilities} of client $j$ and denote it $D_j$.
Observe that $C_j \cap D_j = \emptyset$ for each client $j$. 

We are interested in average distances from a client $j$ to sets of facilities fractionally serving it.
Let $d_j$ be the average connection cost in $x^*_{ij}$ defined as
\[
d_j = d(j, \F) = \frac{\sum_{i \in \F} c_{ij} \cdot x^*_{ij} }{\sum_{i \in \F} x^*_{ij}}  = \sum_{i \in \F} c_{ij} \cdot x^*_{ij}.
\]
Let $d^{(c)}_j$, $d^{(d)}_j$  be the average distances to close and distant facilities defined as
\[
d^{(c)}_j = d(j,C_j) = \frac{\sum_{i \in \F} c_{ij} \cdot \overline{x}_{ij} }{\sum_{i \in \F} \overline{x}_{ij}}
= \frac{\sum_{i \in C_j} c_{ij} \cdot \overline{y}_{i} }{\sum_{i \in C_j} \overline{y}_{i}}  = \sum_{i \in C_j} c_{ij} \cdot \overline{y}_{i},
\]
\[
d^{(d)}_j = d(j,D_j) = \frac{\sum_{i \in D_j} c_{ij} \cdot x^*_{ij} }{\sum_{i \in D_j} x^*_{ij}}
= \frac{\sum_{i \in D_j} c_{ij} \cdot \overline{y}_{i} }{\sum_{i \in D_j} \overline{y}_{i}}
 = \frac{\sum_{i \in D_j} c_{ij} \cdot \overline{y}_{i} }{\gamma - 1}.
\]
Let $d^{(max)}_j = \max_{i \in C_j} c_{ij}$ be the maximum distance of
$j$ to its close facilities, as mentioned above. \\
Let $\rho_j$ be defined as
$
 \rho_j = \frac{d_j - d^{(c)}_j}{d_j}
$
if $d_j > 0$, and define $\rho_j = 0$ otherwise.
Observe that $\rho_j$ takes value between $0$ and $1$.
$\rho_j = 0$ implies $d^{(c)}_j = d_j = d^{(d)}_j$, and
$\rho_j = 1$ occurs only when $d^{(c)}_j = 0$.  Large values of the paprameter$\rho_j$
are desirable for one part of the analysis, and small values are good for 
another part. 

\begin{lemma} \label{average_lemma}
 $d^{(d)}_j = d_j (1+ \frac{\rho_j}{\gamma -1})$.
\end{lemma}

\begin{figure}[t]
\begin{center}
\fboxsep7pt
\framebox[0.95\columnwidth]{
\begin{minipage}{0.90\columnwidth}
\begin{center}

\psfrag{g}{$\frac{1}{\gamma}$}
\psfrag{r}{$1$}
\psfrag{c1}{$d^{(d)}_j$} 
\psfrag{c2}{\hspace{-1mm}$d^{(max)}_j$} 
\psfrag{c3}{$d_j$} 
\psfrag{c4}{\hspace{-6mm}$d^{(c)}_j = d_j(\rho_j - 1)$} 

\includegraphics[width=10cm]{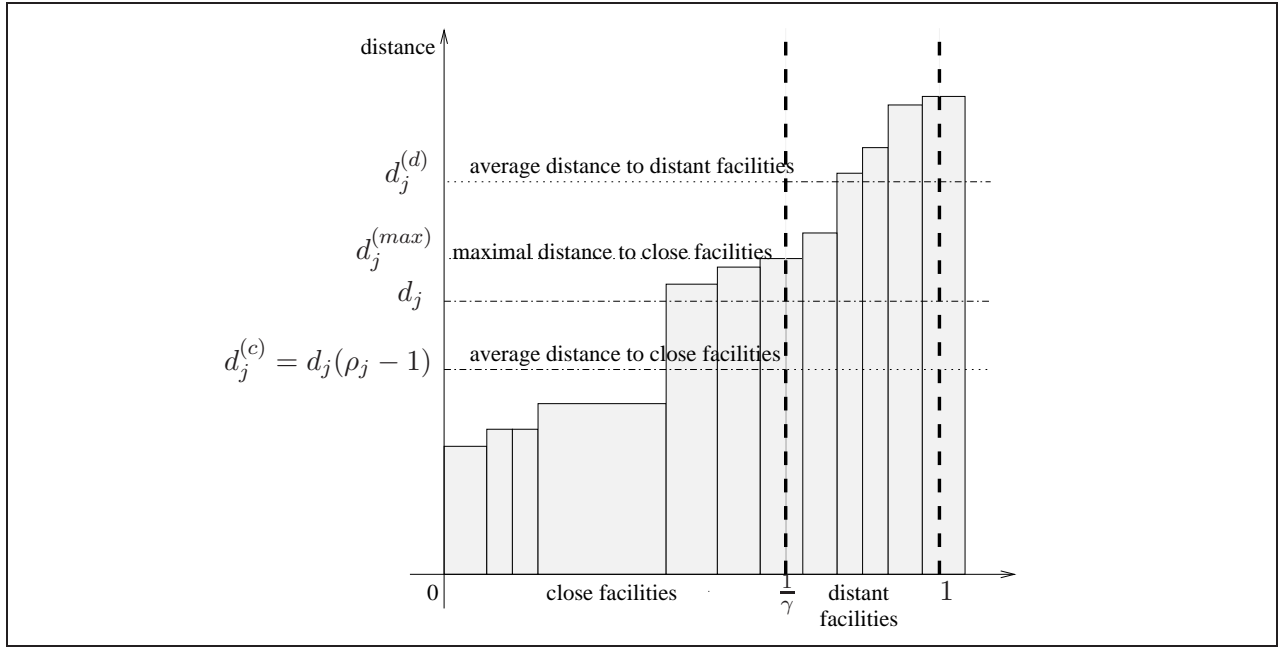}
\end{center}
\end{minipage}
}
\end{center}
\caption{Distances to facilities serving client $j$ in $(x^*,y^*)$. The width of a rectangle corresponding to facility $i$
is equal to $\gamma \cdot x^*_{ij} = \overline{y}_i$. This figure helps 
us understand the meaning of $\rho_j$.}
\label{fig:distances}
\end{figure}

Observe that $d^{(max)}_j \leq d^{(d)}_j$.
We will also use the following lemmas from~\cite{byrka:2212}: 

\begin{lemma} \label{lem:main-lemma}
Suppose $\gamma < 2$ and that clients $j,j' \in \mathcal{C}$ share a facility in $(\overline{x},\overline{y})$, 
i.e., $\exists i \in \mathcal{F} \mbox{ s.t. } \overline{x}_{ij} > 0 \mbox{ and }
\overline{x}_{ij'} > 0 $. Then, either $C_{j'} \setminus
(C_j \cup D_j ) = \emptyset $ or
\[
 d(j,C_{j'} \setminus (C_j \cup D_j )) \leq
 d_{j}^{(d)} + d_{j'}^{(max)} + d_{j'}^{(c)}.
\]
\end{lemma}

\begin{lemma} \label{exp_distance_lemma}
 Given are a random vector $y \in \{0,1\}^{|\mathcal{F}|}$ produced by CS Algorithm
 $CS(\gamma)$, a subset $A\subseteq \mathcal{F}$ of facilities such that $\sum_{i\in A}\bar{y}_i>0$,
 and a client $j\in\mathcal{C}$. Then, the following holds:
\[
 E\left[\min_{i\in A, y_i=1}c_{ij} \ | \ \sum_{i\in A} y_i \geq 1\right] \leq d(j,A)
\]
\end{lemma}

In the analysis of our algorithm we will also use the following result:
\begin{lemma} \label{probability_lemma}
Suppose we are given $n$ independent events that occur with probabilities $p_1, p_2, \ldots, p_n$ respectively.
The probability that at least one of these events occurs is lower-bounded by
$1 - e^{-\sum_{i=1}^n p_i}$.
\end{lemma}

Let $\gamma_0$ be defined as the only positive solution to the following equation.
\begin{equation} \label{eq:gamma0}
\left( \frac{e^{-1} + e^{-\gamma}}{1 - \frac{1}{\gamma}} \right) -  \left(1 + 2e^{-\gamma} \right) = 0 
\end{equation}
An approximate value of this constant is $\gamma_0 \approx 1.67736$.
As we will observe in the proof of Theorem~\ref{thm:main-theorem},
equation~(\ref{eq:gamma0}) appears naturally in the analysis of algorithm $CS(\gamma)$.

\begin{theorem} \label{thm:main-theorem}
For $1 \leq \gamma < 2$, Algorithm $CS(\gamma)$ produces a solution with expected costs
\begin{eqnarray*}
E[cost(OPEN_i)] & = &\gamma \cdot F^*_i, \\
E[cost(CONN_j)] & \leq & \max\left\{1 + 2e^{-\gamma}, \frac{e^{-1}+e^{-\gamma}}{1 - \frac{1}{\gamma}}\right\} \cdot C^*_j\\  
\end{eqnarray*}
where $F^*_i = f_i y^*_i$, $C^*_j = \sum_{i \in \mathcal{F}}{c_{ij} x^*_{ij}}$, 
$F^* = \sum_{i \in \mathcal{F}}{F^*_i}$ and $C^* = \sum_{j \in \mathcal{C}}{C^*_j}$.
\end{theorem}

\begin{proof}
The expected facility-opening cost is 
$E[cost(OPEN_i)] = f_i \overline{y}_i = \gamma  f_i y_i^* = 
\gamma \cdot F^*_i$. 

To bound the expected connection cost, we show that for each client $j$ there is an open facility
within a certain distance with a certain probability.
If $j$ is a cluster center, one of its close facilities is open and the expected distance
to this open facility is $d_j^{(c)}$.

\begin{figure}[t]
\begin{center}
\fboxsep7pt
\framebox[0.95\columnwidth]{
\begin{minipage}{0.90\columnwidth}
\begin{center}

\includegraphics{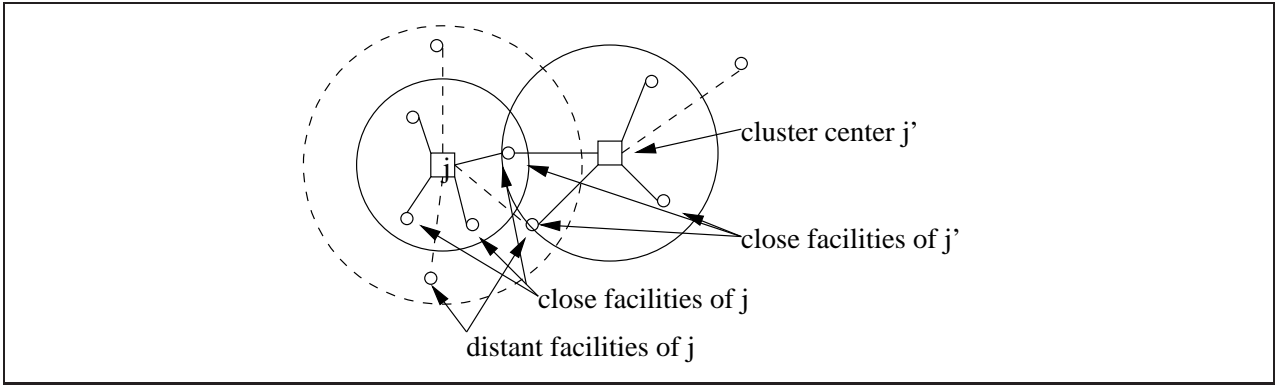}
\end{center}
\end{minipage}
}
\end{center}
\caption{Facilities that $j$ considers: close and distant, as well as close
facilities of its cluster center $j'$.}
\label{fig:routing}
\end{figure}

If $j$ is not a cluster center, it first considers its close facilities (see Figure~\ref{fig:routing}).
If any of them is open, by Lemma~\ref{exp_distance_lemma} the expected distance to the closest open facility is at most $d_j^{(c)}$.
From Lemma~\ref{probability_lemma}, at least one close facility is open with probability $p_c \geq (1-\frac{1}{e})$.
Suppose none of the close facilities of $j$ is open, but at least one of its distant facilities is open.
Let $p_d$ denote the probability of this event. Again by Lemma~\ref{exp_distance_lemma}, the expected distance to the closest facility is then at most $d_j^{(d)}$.
If neither any close nor any distant facility of client $j$ is open,
then $j$ connects itself to the facility serving its cluster center $j'$.
Again from Lemma~\ref{probability_lemma}, such an event happens with probability $p_s \leq \frac{1}{e^\gamma}$.
We will now use the fact that if $\gamma < 2$ then, by Lemma~\ref{lem:main-lemma} and Lemma~\ref{exp_distance_lemma},
the expected distance from $j$ to the facility opened around $j'$ is at most
$d_{j}^{(d)} + d_{j'}^{(max)} + d_{j'}^{(c)}$.

Finally, we combine the probabilities of particular cases
with the bounds on the expected connection for each of the cases, 
to obtain the following upper bound on the expected connection cost.
\begin{eqnarray*}
  E[cost(CONN_j)] & \leq &p_c \cdot d_j^{(c)} + p_d \cdot d_j^{(d)} + p_s \cdot (d_{j}^{(d)} + d_{j'}^{(max)} + d_{j'}^{(c)})\\
  & \leq & \left( (p_c + p_s) \cdot d_j^{(c)} + (p_d + 2p_s) \cdot d_j^{(d)} \right)\\
  & = & \left( (p_c + p_s) \cdot (1- \rho_j) d_j + (p_d + 2p_s) \cdot (d_j (1+ \frac{\rho_j}{\gamma -1})) \right)\\
  & = &((p_c + p_d + p_s) + 2p_s) \cdot d_j + d_j \cdot \left( -\rho_j(p_c + p_s) + (p_d + 2p_s) \cdot \frac{\rho_j}{\gamma -1} \right) \\
  & = &(1 + 2p_s) \cdot d_j +  d_j \cdot \rho_j \left( \frac{1-\gamma p_c + (2 - \gamma)p_s}{\gamma-1} \right)\\
  & \leq &(1 + 2e^{-\gamma}) \cdot d_j +  d_j \cdot \rho_j \left( \frac{1-\gamma (1-e^{-1}) + (2 - \gamma)e^{-\gamma}}{\gamma-1} \right)\\
  & = &(1 + 2e^{-\gamma}) \cdot d_j +  d_j \cdot \rho_j \left( \left( \frac{e^{-1}+e^{-\gamma}}{1 - \frac{1}{\gamma}}\right) - \left(1 + 2e^{-\gamma} \right) \right)\\
  & = &d_j \cdot \left[ (1 - \rho_j)\left(1 + 2e^{-\gamma} \right) + \rho_j \left( \frac{e^{-1}+e^{-\gamma}}{1 - \frac{1}{\gamma}} \right) \right] \\
  & \leq &d_j \cdot \max\left\{1 + 2e^{-\gamma}, \frac{e^{-1}+e^{-\gamma}}{1 - \frac{1}{\gamma}}\right\} \\
  & = &C^*_j \cdot \max\left\{1 + 2e^{-\gamma}, \frac{e^{-1}+e^{-\gamma}}{1 - \frac{1}{\gamma}}\right\}
\end{eqnarray*}


 The penultimate line above is due to the fact that $ 0 \leq \rho_j \leq 1$. 
 The total cost follows easily. Therefore, $CS(\gamma)$ is a 
 $\left(\gamma, \max\left\{1 + 2e^{-\gamma}, \frac{e^{-1}+e^{-\gamma}}{1 - \frac{1}{\gamma}}\right\} \right)$ 
 bi-factor approximation for UFL ($1 \leq \gamma \leq 2$.) 
 Note that $\left(\gamma, 1 + 2e^{-\gamma} \right)$ is the bi-factor approximation lower bound~\cite{DBLP:journals/jacm/JainMMSV03,DBLP:journals/jal/GuhaK99}.
%
%
\end{proof}

\begin{corollary}
The CS(1.575) algorithm is an $1.575$-approximation algorithm
for the UFL problem. 
Also, $CS(\gamma)$ is an optimal bi-factor approximation for UFL for $\gamma_0 < \gamma < 2$.
\end{corollary}

\subsection{Approaching an optimal approximation}

Our purely primal-based analysis of $CS(\gamma)$ suggests a way to 
approach the optimal $s_0$-approximation. Our analysis shows that
if the adversary knows the value of $\gamma$, then the
adversary's optimal strategy -- in selecting a ``bad'' instance
and corresponding LP solution $(x^*, y^*)$ -- is just to make one
of two choices, as follows. For each $j$,
\begin{itemize}
\item \emph{either} all facilities $i$ with $x_{ij}^* > 0$ are at the
same distance from $j$,
\item \emph{or} there exist $0 < a_j < b_j$ such that all facilities $i$ 
with $x_{ij}^* > 0$ are at distance either $a_j$ or $b_j$ from $j$, 
such that $\sum_{i:~c_{ij} = a_j} x_{ij}^*$ is infinitesimally smaller
than $1/\gamma$.
\end{itemize}
Given $\gamma$, the adversary will make the above choice in order to
maximize the expected objective-function value after running $CS(\gamma)$.
(This optimization, as well as the corresponding choices for $a_j$ and $b_j$, 
can be carried out explicitly in a straightforward manner.) 

However, what  if we select $\gamma$ \emph{randomly} from some appropriate
distribution? If the adversary selects the second choice above, then a
suitably ``large'' value of $\gamma$ will defeat the adversary's goal of
making $\sum_{i:~c_{ij} = a_j} x_{ij}^*$ (infinitesimally) smaller
than $1/\gamma$, leading to a significantly-improved approximation. On
the other hand, if the adversary selects the first choice above, then 
a choice of $\gamma$ close to $s_0$ will lead to an approximation close
to $s_0$. 

We are unable at the moment to carry out the above idea in an optimal manner:
choosing the optimal distribution for $\gamma$, for instance. However, we view
this as a potentially-fruitful approach since it starts with the knowledge 
of what exactly are the limits of deterministic strategies for choosing 
$\gamma$.

\section{Two-stage stochastic facility location}
We now consider two-stage stochastic facility location in the
``explicitly-given polynomially-many scenarios'' model presented in the
introduction. The LP-rounding problem is thus to find 
integral solutions ``close'' to the optimal solution of the following LP, 
where the scenarios are indexed by the symbol $A$ and their respective
probabilities aregiven by the values $p_A$: 

\[ \mbox{minimize} ~
\sum_{i \in \mathcal{F}} f_i^I y_i +
\sum_A p_A (\sum_i f_i^A y_{A,i} + \sum_{j \in A} \sum_i c_{ij} x_{A,ij})~
\mbox{subject to} \]
\begin{eqnarray*}
\sum_i x_{A,ij} & \geq & 1 ~~\forall A ~\forall j \in A; 
\label{eqn:facloc-assign} \\
x_{A,ij} & \leq & y_i + y_{A,i} ~~\forall i ~\forall A ~\forall j \in A; 
\label{eqn:facloc-xleqy} \\
x_{A,ij}, y_i, y_{A,i} & \geq & 0 ~~\forall i ~\forall A ~\forall j \in A.
\nonumber
\end{eqnarray*}

First, in Section~\ref{section:general_exp} we give an algorithm for the standard setting,
then in Section~\ref{section:per_scenario}  we consider the per-scenario version of the problem.

\subsection{General expected cost: analysis with a dual bound}
\label{section:general_exp}
Consider the following dual formulation of the 2-stage stochastic facility location problem: 
\begin{eqnarray*}
 \mbox{maximize} && \sum_{A} p_A (\sum_{j \in A} v_{j,A})
~\mbox{subject to:}\\
 && w_{ij,A} \geq v_{j,A} - c_{ij}\\
 && f_i^I \leq \sum_{A} p_A (\sum_{j \in A} w_{ij,A})\\
 && f_i^A \leq \sum_{j \in A} w_{ij,A}\\
 &&  w_{ij,A},v_{j,A} \geq 0.
\end{eqnarray*}

Let $(x^*,y^*)$ and $(v^*,w^*)$ be optimal solutions to the primal
and the dual programs, respectively. Note that by complementary 
slackness, we have
$c_{ij} \leq v_{j,A}$ if $x_{A,ij} > 0$. 

\paragraph{Algorithm.}

We now describe a randomized LP-rounding algorithm that transforms the fractional 
solution $(x^*,y^*)$ into an integral solution $(\hat{x},\hat{y})$ with bounded expected cost.
The expectation is over the random choices of the algorithm, but not over the
random choice of the scenario. Note that we need to decide the first stage entries of $\hat{y}$
not knowing $A$. W.l.o.g.\ we assume that no facility is fractionally opened in $(x^*,y^*)$ 
in both stages, i.e, for all $i$ we have $y^*_i = 0$ or for all $A$ $y^*_{A,i} = 0$.
To obtain this property it suffices to have two identical copies of each facility,
one for Stage I and one for Stage II.

We start by scaling the fractional solution $(x^*,y^*)$ by a factor of 2.
As a result, we obtain a fractional solution $(\overline{x},\overline{y})$
with $\overline{x}_{A,ij} = 2 \cdot x^*_{A,ij} $, $\overline{y}_i =  2 \cdot y^*_i $,
and $\overline{y}_{A,i} =  2 \cdot y^*_{A,i} $. Note that the scaled fractional solution
$(\overline{x},\overline{y})$ can have facilities with fractional opening of more then 1.
For simplicity of the analysis, we do not round these facility-opening values to 1, but rather split such facilities.
More precisely, we split each facility $i$ with fractional opening $\overline{y}_i > \overline{x}_{A,ij} > 0$ (or $\overline{y}_{A,i} > \overline{x}_{A,ij} > 0$) for some $(A, j)$ into $i'$ and $i''$, such that 
 $\overline{y}_{i'} = \overline{x}_{A,ij}$ and $\overline{y}_{i''} = \overline{y}_i - \overline{x}_{A,ij}$.
We also split facilities whose fractional opening exceeds one.
By splitting facilities we create another instance of the problem, then we solve this modified instance
and interpret the solution as a solution to the original problem in the natural way.
The technique of splitting facilities is precisely described in~\cite{DBLP:conf/ipco/Sviridenko02}.

Define $\overline{x}^{(I)}_{A,ij} = \min \{\overline{x}_{A,ij}, \overline{y}_i \}$, and 
$\overline{x}^{(II)}_{A,ij} = \overline{x}_{A,ij} - \overline{x}^{(I)}_{A,ij}$.
Observe, that for a client-scenario pair $(j,A)$ either $\sum_{i \in \F} \overline{x}^{(I)}_{A,ij} \geq 1$, or
$\sum_{i \in \F} \overline{x}^{(II)}_{A,ij} > 1$. In the former case, we call such a pair \emph{first stage served},
and we denote the set of the first stage served pairs by $S$.

Since we can split facilities, for each $(j,A) \in S$ we can assume that
there exists a subset of facilities $F_{(j,A)} \subseteq \F$, 
such that $\sum_{i \in  F_{(j,A)} } \overline{x}^{(I)}_{A,ij} = 1$,
and for each $i \in F_{(j,A)}$ we have $\overline{x}^{(I)}_{A,ij} = \overline{y}_i$.
Also for each $(j,A) \notin S$ we can assume that
there exists a subset of facilities $F_{(j,A)} \subseteq \F$,
such that $\sum_{i \in  F_{(j,A)} } \overline{x}^{(II)}_{A,ij} = 1$,
and for each $i \in F_{(j,A)}$ we have $\overline{x}^{(II)}_{A,ij} = \overline{y}_{A,i}$.
Let $R_{(j,A)} = \max_{i \in F_{(j,A)}} c_{ij}$ be a maximal distance from $j$ to an $i \in F_{(j,A)}$.
Recall that, by complementary slackness, we have $R_{(j,A)} \leq v_{j,A}$. 

The algorithm opens facilities randomly in each of the stages
with the probability of opening facility $i$ equal to $\overline{y}_i$
in Stage I, and $\overline{y}_{A,i}$ in Stage II of scenario $A$.
Some facilities are grouped in disjoint \emph{clusters} in order to correlate 
the opening of facilities from a single cluster.
The clusters are formed in each stage by the following procedure.
Let all facilities be initially unclustered.
In Stage I, consider all client-scenario pairs $(j,A) \in S$
(in Stage II of scenario $A$, consider all clients $j$ such that $(j,A) \notin S$)
in the order of non-decreasing values $R_{(j,A)}$.
If the set of facilities $F_{(j,A)}$ contains no facility
from the previously formed clusters, 
then form a new cluster containing facilities from $F_{(j,A)}$,
otherwise do nothing.
In each stage, open exactly one facility in each cluster.
Recall that the total fractional opening of facilities in each cluster
equals 1. Within each cluster choose the facility randomly with the probability
of opening facility $i$ equal to the fractional opening $\overline{y}_i$ 
in Stage I, or $\overline{y}_{A,i}$ in Stage II of scenario $A$.
For each unclustered facility $i$ open it independently with 
probability $\overline{y}_i$ in Stage I, and with probability $\overline{y}_{A,i}$ in Stage II of scenario $A$.
Finally, at the end of Stage II of scenario $A$, connect each client $i \in A$ to the 
closest open facility.

%

\paragraph{Analysis.}
Consider the solution $(\hat{x},\hat{y})$ constructed by our LP-rounding algorithm.
We fix scenario $A$ and bound the expectation of 
$COST(A) = \sum_{i \in \F} (f_i^I \hat{y}_{i} + f_i^A \hat{y}_{A,i}) + \sum_{j\in A}\sum_{i\in \F} c_{ij}\hat{x}_{A,ij}$.
Define $C_A = \sum_{j \in A} C_{(j,A)} = \sum_{j \in A} \sum_{i \in \mathcal{F}}c_{ij}x^*_{A,ij}$,
$F_A = \sum_{i \in \F} (f_i^I y^*_{i} + f_i^A y^*_{A,i})$, $V_A = \sum_{j \in A} v^*_{j,A}$.

\begin{lemma} \label{primal_dual_lemma}
  $E[COST(A)] \leq e^{-2} \cdot 3 \cdot V_A + (1-e^{-2}) \cdot C_A + 2 \cdot F_A$ in each scenario $A$.
\end{lemma}

\begin{proof}
Since the probability of opening a facility is
equal to its fractional opening in $(\overline{x},\overline{y})$, 
the expected facility-opening cost of $(\hat{x},\hat{y})$ 
equals facility-opening cost of $(\overline{x},\overline{y})$, which
is exactly twice the facility-opening cost of $(x^*,y^*)$.

Fix a client $j \in A$. 
The total (from both stages) fractional opening in $\overline{y}$ of facilities serving $j$ in 
$(\overline{x},\overline{y})$ is exactly 2, hence the probability that at least one of these 
facilities is open in $(\hat{x},\hat{y})$ is at least $1-e^{-2}$.
Observe that, on the condition that at least one such facility is open,
by an analogous to the one from Lemma~\ref{exp_distance_lemma}, 
the expected distance to the closest of the open facilities is at most $C_{(j,A)}$.

With probability at most $e^{-2}$, none of the facilities fractionally serving $j$
in $(\overline{x},\overline{y})$ is open. In such a case we need to find a different
facility to serve $j$. We will now prove that for each client $j \in A$ there exists a facility
$i$ which is open in $(\hat{x},\hat{y})$, such that $c_{ij} \leq 3 \cdot v_{j,A}$.

Assume $(j,A) \in S$ (for $(j,A) \notin S$ the argument is analogous). 
If $F_{(j,A)}$ is a cluster, then at least one $i \in F_{(j,A)}$ is open
and $c_{ij} \leq v_{j,A}$.
Suppose $F_{(j,A)}$ is not a cluster, then by the construction of clusters,
it intersects a cluster $F_{(j',A')}$ with $R_{(j',A')} \leq R_{(j,A)} \leq v_{j,A}$.
Let $i$ be the facility opened in cluster $F_{(j',A')}$ and let $i' \in F_{(j',A')} \cap F_{(j,A)}$.
Since $i'$ is in $F_{(j,A)}$, $c_{i'j} \leq R_{(j,A)}$. Since both $i$ and $i'$ are in $F_{(j',A')}$,
both $c_{ij'} \leq R_{(j',A')}$ and $c_{i'j'} \leq R_{(j',A')}$.
Hence, by triangle inequality, $c_{ij} \leq R_{(j,A)} + 2 \cdot R_{(j',A')} \leq 3 \cdot R_{(j,A)} \leq 3 \cdot v_{j,A}$.

Thus, the expected cost of the solution in scenario $A$ is:
\begin{eqnarray*} 
 E[COST(A)] & \leq & e^{-2} \cdot 3 \cdot \sum_{j \in A} v_{j,A} + 
 (1-e^{-2})  (\sum_{j\in A}\sum_{i\in \F} c_{ij}x^*_{A,ij}) + 
 2 \cdot (\sum_{i \in \F} (f_i^I y^*_{i} + f_i^A y^*_{A,i})) \\
 & \leq & e^{-2} \cdot 3 \cdot V_A + (1-e^{-2}) \cdot C_A + 2 \cdot F_A.
\end{eqnarray*}
\end{proof}

Define $F^* = \sum_{i \in \mathcal{F}} f_i^I y_i +
\sum_A p_A (\sum_i f_i^A y_{A,i})$
and $C^* = \sum_A p_A (\sum_{j \in A} \sum_i c_{ij} x_{A,ij})$.
Note that we have $F^* = \sum_A p_A F_A$, $C^* = \sum_A p_A C_A$, and
$F^*+C^* = \sum_A p_A V_A$.
Summing up the expected cost over scenarios we obtain the following
estimate on the general expected cost, where the expectation is both 
on the choice of the scenario and on the random choices of our algorithm.

\begin{corollary}
 $E[COST(\hat{x},\hat{y})] \leq 2.4061 \cdot F^* + 1.2707 \cdot C^*$.
\end{corollary}

\begin{proof}
\begin{eqnarray*}
 E[COST(\hat{x},\hat{y})]& = & \sum_A p_A E[COST(A)]\\
 & \leq & \sum_A p_A\left( e^{-2} \cdot 3 \cdot V_A + (1-e^{-2}) \cdot C_A + 2 \cdot F_A \right)\\
 & = & (1-e^{-2}) \cdot C^* + 2 \cdot F^* + 3e^{-2} (\sum_A p_A V_A)\\
 & = & (1-e^{-2}) \cdot C^* + 2 \cdot F^* + 3e^{-2} (F^*+C^*)\\
 & = & (2 + e^{-2} \cdot 3) F^* + (1 + e^{-2} \cdot 2) C^* \\
 & \leq & 2.4061 \cdot F^* + 1.2707 \cdot C^*.
\end{eqnarray*}
\end{proof}

\paragraph{Combining two algorithms.}
The above described algorithm can be combined with an algorithm from~\cite{DBLP:conf/soda/Srinivasan07}
to obtain a $2.2975$-approximation algorithm.
See Appendix~\ref{app:comb-alg} for details.


\subsection{Per-scenario bounds: primal analysis}
\label{section:per_scenario}

Consider again the $2$-stage facility location problem, and a corresponding
optimal fractional solution. We now describe a randomized rounding
scheme so that for each scenario $A$, its expected final (rounded)
cost is at most $2.4061$ times its fractional counterpart 
$Val_A = \sum_{i \in \F} (f_i^I y^*_i + f_i^A y^*_{A,i}) + \sum_{j\in A}\sum_{i\in \F} c_{ij}x^*_{A,ij}$,
improving on the $3.095 \cdot Val_A$
bound of \cite{DBLP:conf/soda/Srinivasan07}.

Note that we cannot use dual bounds in this setting, as the dual budgets $V_A$
do not have to equal $Val_A$ in each scenario (see Appendix~\ref{app:per-scenario} for an example).
Instead, we scale the facility-opening values a little more and show that scaling by a factor of 2.4061 is
sufficient to bound the expected connection cost in each scenario by 2.4061 times the fractional connection
cost in this scenario. Details are given in Appendix~\ref{app:per-scenario}.

\section{Robust fault-tolerant UFL}
In recent work, Chechik and Peleg
have introduced a new variant of facility location 
problems~\cite{DBLP:conf/stacs/ChechikP10}.
They study a setting that can be described as follows. Once we choose the facilities to open, an adversary closes up to $k$ of them
(which models possible failures of facilities), and then clients are connected to the closest of the remaining open facilities. 
The goal is to minimize the facility-opening cost plus the worst-case (over the choice of facilities to close) connection cost.
Observe that integral solutions $(x,y)$ of the following linear program are exactly the feasible solutions
to the problem we study.

\begin{equation} \label{eqn:rf0}
 \mbox{minimize} ~
\sum_{i \in \mathcal{F}} f_i y_i +
\max_A \sum_{j} \sum_i c_{ij} x_{A,ij}~
\mbox{subject to} 
\end{equation}
\begin{eqnarray}
\sum_i x_{A,ij} & \geq & 1 ~~\forall A ~\forall j; 
\label{eqn:rf1} \\
x_{A,ij} & \leq & y_i  ~~\forall i ~\forall A ~\forall j ; 
\label{eqn:rf2} \\
x_{A,ij} & = & 0  ~~\forall A ~\forall i \in A ~\forall j ; 
\label{eqn:rf3} \\
x_{A,ij}, y_i, y_{A,i} & \geq & 0 ~~\forall i ~\forall A ~\forall j.
\nonumber
\end{eqnarray}

The ``scenarios'' $A$ in the above program are all the subsets of facilities of cardinality $k$,
and they encode the facilities closed by the adversary.
Note that the connection cost is calculated as a maximum over the scenarios.
The above program is of polynomial size only for fixed $k$ and we will only study settings with such small $k$.

In~\cite{DBLP:conf/stacs/ChechikP10} Chechik and Peleg gave a 6.5-approximation algorithm for $k=1$ and a ($7.5k + 1.5$)-approximation algorithm
for general $k$. We improve the latter to a $(k+5+4/k)$-approximation algorithm for the studied problem.
This we obtain by showing that scaling the facility opening variables by $(k+5+4/k)$ is sufficient
to provide enough fractional opening, which is then rounded by a dependent rounding method 
descrikbed in~\cite{ipco10} (Sections 3 and 4). More details are given in Appendix~\ref{app:rft}.

We also briefly discuss an oblivious version of the problem, where the adversary does not know
the random choices of the algorithm when deciding the facilities to close. In this setting 
we provide a randomized LP-rounding algorithm that delivers integral solutions
of expected cost at most $k + 1.5$ times the cost of the initial fractional solution. (see Appendix~\ref{app:rft})

We also show that these methods cannot be extended to obtain approximation ratios sublinear 
in $k$ by providing instances with integrality gap that are arbitrarily
close to $k + 1$. (see Appendix~\ref{app:rft})

\bibliographystyle{abbrv}
\bibliography{fls}

\begin{thebibliography}{10}

\bibitem{AS}
A.~Ageev and M.~Sviridenko.
\newblock Pipage rounding: a new method of constructing algorithms with proven
  performance guarantee.
\newblock {\em Journal of Combinatorial Optimization}, 8(3):307--328, 2004.

\bibitem{byrka:2212}
J.~Byrka and K.~Aardal.
\newblock An optimal bifactor approximation algorithm for the metric
  uncapacitated facility location problem.
\newblock {\em SIAM Journal on Computing}, 39(6):2212--2231, 2010.

\bibitem{ipco10}
J.~Byrka, A.~Srinivasan, and C.~Swamy.
\newblock Fault-tolerant facility location: a randomized dependent
  {LP}-rounding algorithm.
\newblock In {\em IPCO}, pages 244--257, 2010.

\bibitem{DBLP:conf/stacs/ChechikP10}
S.~Chechik and D.~Peleg.
\newblock Robust fault tolerant uncapacitated facility location.
\newblock In {\em STACS}, pages 191--202, 2010.

\bibitem{DBLP:journals/siamcomp/ChudakS03}
F.~A. Chudak and D.~B. Shmoys.
\newblock Improved approximation algorithms for the uncapacitated facility
  location problem.
\newblock {\em SIAM J. Comput.}, 33(1):1--25, 2003.

\bibitem{DBLP:journals/jal/GuhaK99}
S.~Guha and S.~Khuller.
\newblock Greedy strikes back: Improved facility location algorithms.
\newblock {\em J. Algorithms}, 31(1):228--248, 1999.

\bibitem{DBLP:journals/jacm/JainMMSV03}
K.~Jain, M.~Mahdian, E.~Markakis, A.~Saberi, and V.~V. Vazirani.
\newblock Greedy facility location algorithms analyzed using dual fitting with
  factor-revealing {LP}.
\newblock {\em J. ACM}, 50(6):795--824, 2003.

\bibitem{DBLP:conf/stoc/LinV92}
J.-H. Lin and J.~S. Vitter.
\newblock $\epsilon$-approximations with minimum packing constraint violation
  (extended abstract).
\newblock In {\em STOC}, pages 771--782, 1992.

\bibitem{DBLP:journals/siamcomp/MahdianYZ06}
M.~Mahdian, Y.~Ye, and J.~Zhang.
\newblock Approximation algorithms for metric facility location problems.
\newblock {\em SIAM J. Comput.}, 36(2):411--432, 2006.

\bibitem{DBLP:conf/stoc/ShmoysTA97}
D.~B. Shmoys, {\'E}.~Tardos, and K.~Aardal.
\newblock Approximation algorithms for facility location problems (extended
  abstract).
\newblock In {\em STOC}, pages 265--274, 1997.

\bibitem{DBLP:conf/soda/Srinivasan07}
A.~Srinivasan.
\newblock Approximation algorithms for stochastic and risk-averse optimization.
\newblock In {\em SODA}, pages 1305--1313, 2007.

\bibitem{DBLP:conf/ipco/Sviridenko02}
M.~Sviridenko.
\newblock An improved approximation algorithm for the metric uncapacitated
  facility location problem.
\newblock In {\em IPCO}, pages 240--257, 2002.

\bibitem{DBLP:conf/focs/SwamyS05}
C.~Swamy and D.~B. Shmoys.
\newblock Sampling-based approximation algorithms for multi-stage stochastic
  optimization.
\newblock In {\em FOCS}, pages 357--366, 2005.

\end{thebibliography}

\appendix
\begin{center}
\Large{\textbf{Appendix}}
\end{center}

\section{Combining two algorithms for 2-stage stochastic facility location}
\label{app:comb-alg}

We have described an algorithm that returns solutions of expected cost at most $2.4061 \cdot F^* + 1.2707 \cdot C^*$.
Let us call this algorithm ALG1.

In~\cite{DBLP:conf/soda/Srinivasan07}, Srinivasan gave a different approximation algorithm for our 2-stage stochastic facility location problem.
This algorithm also splits the client-scenario pairs into two groups, namely those to be connected in the first stage,
and those that are left to be connected in the second stage. The decision is made by comparing fractional ``first stage'' connection 
of each pair with a certain threshold. Once the split is made, the obtained instances of the standard Uncapacitated Facility Location problem are
solved with the JMS algorithm~\cite{DBLP:journals/jacm/JainMMSV03}. 
The threshold is chosen randomly from a distribution parametrized 
by $\alpha$. For the choice of a parameter $\alpha = 0.2485$ the resulting algorithm is shown in~\cite{DBLP:conf/soda/Srinivasan07} to be a $2.369$-approximation algorithm.
It is easy to show that by setting $\alpha = 0.37$ in the algorithm of~\cite{DBLP:conf/soda/Srinivasan07}, we obtain
an algorithm that returns solutions of expected cost of at most $2.24152 F^* + 2.8254 C^*$.
We will call this algorithm ALG2.

Consider the algorithm ALG3, which tosses a coin that comes heads with probability $p=0.3396$.
If the coin comes heads, then ALG1 is executed; if it comes tails ALG2 is used.
The expected cost of the solution produced by ALG3 can be estimated as: $ F^* (p \cdot 2.4061 + (1-p)\cdot 2.24152) + 
C^*(p \cdot 1.2707 + (1-p)\cdot 2.8254) \leq 2.2975 (C^*+F^*)$.
Therefore, ALG3 is a $2.2975$-approximation algorithm for the 2-stage stochastic facility location problem.

\section{2-stage stochastic facility location with per-scenario bouds}
\label{app:per-scenario}

Let us first note that it is not possible to directly use the analysis from
the previous setting in the per-scenario model. This is because the dual costs $V_A$
do not need to be equal $Val_A = F_A + C_A$ in each scenario $A$. It is possible, for instance,
that the fractional opening of a facility in the first stage is entirely paid 
form the dual budget of a single scenario, despite the fact that clients not active
in this scenario benefit from the facility being open. This can be observed, e.g., in the following
simple example. 
Consider two clients $c^1$ and $c^2$, and two facilities $f^1$ and $f^2$.
All client facility distances are 1, except $c_{1,2}=dist(c^1,f^2)=3$.
Scenarios are: $A^1 = \{c^1\}$ and $A^2=\{c^2\}$, and they occur with probability 1/2 each.
The facility-opening costs are: $f_1^I=2$, $f_2^I=\epsilon$, $f_1^A=f_2^A=4$ for both scenarios $A$. 
It is easy to see that the only optimal fractional solution is integral
and it opens facility $f^1$ in the first stage, and opens no more facilities
in the second stage. Therefore, $Val(A^1)=Val(A^2)=3$.
However, in the dual problem, client $c^2$ has an advantage over $c^1$ in the access 
to the cheaper facility $f^2$, and therefore in no optimal dual solution client $c_2$ will pay
more then $\epsilon$ for the opening of facility $f^1$. In consequence, most of the cost of opening $f^1$
is paid by the dual budget of scenario $A^1$. Therefore, the dual budget $V_{A^1}$ 
is strictly greater then the primal bound $Val_{A^1}$ which we use as an estimate
of the cost of the optimal solution in scenario $A^1$.

Bearing the above example in mind, we construct an LP-rounding algorithm
that does not rely on the dual bound on the length of the created connections.
we use a primal bound, which is obtained by scaling the opening variables a little more
and using just a subset of fractionally connected facilities for each client
in the process of creating clusters. Such a simple filtering technique,
whose origins can be found in the work of Lin and Vitter~\cite{DBLP:conf/stoc/LinV92}, provides slightly weaker
but entirely primal, per-scenario bounds.

\paragraph{Algorithm.}

As before, we describe a randomized LP-rounding algorithm that transforms the fractional 
solution $(x^*,y^*)$ into an integral solution $(\hat{x},\hat{y})$ with bounded expected cost.
The expectation is over the random choices of the algorithm, but not over the
random choice of the scenario. 

We start by scaling the fractional solution $(x^*,y^*)$ by a factor of $\gamma > 2$.
As a result, we obtain a fractional solution $(\overline{x},\overline{y})$
with $\overline{x}_{A,ij} = \gamma \cdot x^*_{A,ij} $, $\overline{y}_i =  \gamma \cdot y^*_i $,
and $\overline{y}_{A,i} =  \gamma \cdot y^*_{A,i} $. Note that the scaled fractional solution
$(\overline{x},\overline{y})$ may have facilities with fractional opening of more than $1$.
Again, for simplicity of the analysis, we do not round these facility-opening values to 1, but rather split such facilities.
More precisely, we split each facility $i$ with fractional opening $\overline{y}_i > \overline{x}_{A,ij} > 0$ (or $\overline{y}_{A,i} > \overline{x}_{A,ij} > 0$) for some $(A, j)$ into $i'$ and $i''$, such that 
 $\overline{y}_{i'} = \overline{x}_{A,ij}$ and $\overline{y}_{i''} = \overline{y}_i - \overline{x}_{A,ij}$.
We also split facilities whose fractional opening exceeds one.

As before, define $\overline{x}^{(I)}_{A,ij} = \min \{\overline{x}_{A,ij}, \overline{y}_i \}$, and 
$\overline{x}^{(II)}_{A,ij} = \overline{x}_{A,ij} - \overline{x}^{(I)}_{A,ij}$.
Define \\ 
\[
F_{(j,A)}^I = \left\{ 
\begin{array}{l l}
    argmin_{F' \subseteq \F : \sum_{i \in F'} \overline{x}^{(I)}_{A,ij} \geq 1 } max_{i \in F'} c_{ij} &  
    \mbox{if $\sum_{i \in \F} \overline{x}^{(I)}_{A,ij} \geq 1$}\\
  \emptyset &  \mbox{if $\sum_{i \in \F} \overline{x}^{(I)}_{A,ij} < 1$}\\
\end{array} \right.
\]

\[
F_{(j,A)}^{II} = \left\{ 
\begin{array}{l l}
    argmin_{F' \subseteq \F : \sum_{i \in F'} \overline{x}^{(II)}_{A,ij} \geq 1 } max_{i \in F'} c_{ij} &  
    \mbox{if $\sum_{i \in \F} \overline{x}^{(II)}_{A,ij} \geq 1$}\\
  \emptyset &  \mbox{if $\sum_{i \in \F} \overline{x}^{(II)}_{A,ij} < 1$}\\
\end{array} \right.
\]

Note that  these sets can easily be computed by considering facilities in an order of non-decreasing distances $c_{ij}$
to the considered client $j$.
Since we can split facilities, w.l.o.g., for all $j \in \C$ we assume that if $F_{(j,A)}^{I}$ is nonempty then
$\sum_{i \in F_{(j,A)}^{I}} \overline{x}^{(I)}_{A,ij} = 1$, and if $F_{(j,A)}^{II}$ is not empty then
$\sum_{i \in F_{(j,A)}^{II}} \overline{x}^{(II)}_{A,ij} = 1$.
Define $d_{(j,A)}^I = max_{i \in F_{(j,A)}^I} c_{ij}$ and $d_{(j,A)}^{II} = max_{i \in F_{(j,A)}^{II}} c_{ij}$.
Let $d_{(j,A)} = min\{d_{(j,A)}^I, d_{(j,A)}^{II}\}$.

For a client-scenario pair $(j,A)$, if we have $d_{(j,A)} = d_{(j,A)}^I$, then we call such a pair \emph{first-stage clustered},
and put its \emph{cluster candidate} $F_{(j,A)} = F_{(j,A)}^I$. Otherwise, if $d_{(j,A)} = d_{(j,A)}^{II} < d_{(j,A)}^{I}$, 
we say that $(j,A)$ is \emph{second-stage clustered} and put its \emph{cluster candidate} $F_{(j,A)} = F_{(j,A)}^{II}$

Recall that we use $C_{(j,A)}=\sum_i c_{ij} x^*_{A,ij}$ to denote the fractional connection cost of client $j$ in scenario $A$.
Let us now argue that distances to facilities in \emph{cluster candidates} are not too large. 

\begin{lemma} \label{lemma:d_j_A}
 $d_{(j,A)} \leq \frac{\gamma}{\gamma-2} \cdot C_{(j,A)}$ for all pairs $(j,A)$. 
\end{lemma}

\begin{proof}
Fix a client-scenario pair $(j,A)$. Assume $F_{(j,A)} = F_{(j,A)}^I$ (the other case is symmetric).
Recall that in this case we have $d_{(j,A)} = d_{(j,A)}^I \leq d_{(j,A)}^{II}$.
Consider the following two subcases.\\
\textbf{Case 1.} $\sum_{i \in F_{(j,A)}^{II}} \overline{x}^{(II)}_{A,ij} = 1$.\\
Observe that we have $c_{ij} \geq d_{(j,A)}$ for all $i \in F'=\F\setminus (F_{(j,A)}^I \cup F_{(j,A)}^{II})$.
Note also that $\sum_{i \in F'} \overline{x}_{A,ij} = \gamma - 2$ and
$\sum_{i \in F'} x^*_{A,ij} = \frac{\gamma - 2}{\gamma}$. 
Hence, $C_{(j,A)}=\sum_{i \in \F} x^*_{A,ij} c_{ij} \geq \sum_{i \in F'} x^*_{A,ij} c_{ij} \geq \frac{\gamma - 2}{\gamma} \cdot d_{(j,A)}$.\\
\textbf{Case 2.} $\sum_{i \in F_{(j,A)}^{II}} \overline{x}^{(II)}_{A,ij} = 0$, which implies that 
$\sum_{i \in \F} \overline{x}^{(II)}_{A,ij} < 1$. Observe that now we have $\sum_{i \in \F} \overline{x}^{(I)}_{A,ij}  > \gamma - 1$,
and therefore $\sum_{i \in \F \setminus F_{(j,A)}^I} \overline{x}^{(I)}_{A,ij} > \gamma-2$.
Recall that $c_{ij} \geq d_{(j,A)}$ for all $i \in (\F\setminus F_{(j,A)}^I) $, hence
$C_{(j,A)} = \sum_{i \in \F} x^*_{A,ij} c_{ij} \geq \sum_{i \in (\F\setminus F_{(j,A)}^I)} x^*_{A,ij} c_{ij} > 
\frac{\gamma - 2}{\gamma} \cdot d_{(j,A)}$. 
\end{proof}
 
Like in Section~\ref{section:general_exp}, the algorithm opens facilities randomly in each of the stages
with the probability of opening facility $i$ equal to $\overline{y}_i$
in Stage I, and $\overline{y}_{A,i}$ in Stage II of scenario $A$.
Some facilities are grouped in disjoint \emph{clusters} in order to correlate 
the opening of facilities from a single cluster.
The clusters are formed in each stage by the following procedure.
Let all facilities be initially unclustered.
In Stage I, consider all first-stage clustered client-scenario pairs, 
i.e., pairs $(j,A)$ such that $d_{(j,A)} = d_{(j,A)}^I$.
(in Stage II of scenario $A$, consider all second-stage clustered client-scenario pairs)
in the order of non-decreasing values $d_{(j,A)}$.
If the set of facilities $F_{(j,A)}$ contains no facility
from the previously formed clusters, 
then form a new cluster containing facilities from $F_{(j,A)}$,
otherwise do nothing.
In each stage, open exactly one facility in each cluster.
Recall that the total fractional opening of facilities in each cluster
equals 1. Within each cluster choose the facility randomly with the probability
of opening facility $i$ equal to the fractional opening $\overline{y}_i$ 
in Stage I, or $\overline{y}_{A,i}$ in Stage II of scenario $A$.
For each unclustered facility $i$ open it independently with 
probability $\overline{y}_i$ in Stage I, and with probability $\overline{y}_{A,i}$ in Stage II of scenario $A$.

Finally, at the end of Stage II of scenario $A$, connect each client $i \in A$ to the 
closest open facility.

\paragraph{Analysis.}
The expected facility-opening cost is obviously $\gamma$ times the fractional opening cost.
More precisely, the expected facility-opening cost in scenario $A$ equals $\gamma \cdot F^*_A =
\gamma \cdot \sum_{i \in \mathcal{F}} f_i^I y_i +\sum_i f_i^A y_{A,i}$
it remains to bound the expected connection cost in scenario $A$ in terms of 
$C^*_A =\sum_{j \in A} \sum_i c_{ij} x_{A,ij}$. 

\begin{lemma}
The expected connection cost in scenario A is at most $(1+\frac{2\gamma+2}{\gamma-2} e^{-\gamma}) \cdot C_{(j,A)}$.
\end{lemma}

\begin{proof}
Consider a single client-scenario pair $(j,A)$.
Observe that the facilities fractionally connected to $j$ in scenario $A$
have the total fractional opening of $\gamma$ in the scaled facility-opening vector $\overline{y}$.
Since there is no positive correlation (only negative correlation in the disjoint clusters formed 
by the algorithm), with probability at least $1-e^{-\gamma}$ at least one such facility will be opened,
moreover, by Lemma~\ref{exp_distance_lemma} the expected distance to the closest of the open facilities from this set
will be at most the fractional connection cost $C_{(j,A)}$.

Just like in the proof of Lemma~\ref{primal_dual_lemma}, from the greedy construction of the clusters
in each phase of the algorithm, with probability 1, there exists facility $i$ opened 
by the algorithm such that $c_{ij} \leq 3 \cdot d_{(j,A)}$.
We connect client $j$ to facility $i$ if no facility from facilities fractionally serving $(j,A)$ was opened.
We obtain that the expected connection cost of client $j$ is at most $(1-e^{-\gamma}) \cdot C_{(j,A)} + e^{-\gamma} \cdot 3 d_{(j,A)}$.
By Lemma~\ref{lemma:d_j_A}, this can by bounded by  
$(1-e^{-\gamma}) \cdot C*_{j,A} + e^{-\gamma} \cdot 3 \cdot \frac{\gamma}{\gamma-2} \cdot C_{(j,A)} =
(1+\frac{2\gamma+2}{\gamma-2} e^{-\gamma}) \cdot C_{(j,A)}$
\end{proof}

To equalize the opening and connection cost approximation ratios we solve $(1+\frac{2\gamma+2}{\gamma-2} e^{-\gamma}) = \gamma$ and obtain the following.

\begin{theorem}
The described algorithm with $\gamma = 2.4957$ delivers solutions such that the expected cost in each scenario $A$
is at most 2.4957 times the fractional cost in scenario $A$. 
\end{theorem}

\section{Robust fault-tolerant UFL}
\label{app:rft}

\subsection{The $(k+5+4/k)$-approximation rounding routine}
\label{sec:rftfl-general}
Like in the algorithms in the previous sections we first scale up the fractional facility-opening costs,
we then cluster certain facilities to correlate their opening,
and then use a randomized rounding routine to decide the subset of facilities to open.
Once we open facilities and the adversary chooses which $k$ of them to close,
clients get connected to the closest of the remaining open facilities.

Let $(x^*,y^*)$ be an optimal solution to the above LP relaxation 
of the problem.
We first scale up the opening of facilities by $\gamma=k+5+4/k$,
i.e., we set $\overline{y}_i = \min\{1,\gamma \cdot y^*_i \}$.
We also set $\overline{x}_{A,ij} = \min\{1,\gamma \cdot x^*_{A,ij} \}$
 
Consider a single client-scenario pair $(j,A)$.
Consider facilities $i$ fractionally serving this pair in solution $(x^*,y^*)$
in an order $i^1, i^2, \ldots $ of non-decreasing distance to $c_{ij}$.
Let $i'$ be the first facility in this order such that 
$x^*_{A,i^1j} + x^*_{A,i^2j} + x^*_{A,i'j} \geq \frac{k+1}{\gamma} = \frac{k}{k+4}$.
Recall that $C_{(j,A)} = \sum_i c_{ij} x^*_{A,ij}$
denotes the fractional connection cost of client $j$ in scenario $A$.
By an argument analogous to the one in Lemma~\ref{exp_distance_lemma}, we obtain that 
$c_{i'j} \leq \frac{(k+4)(k+1)}{3k} \cdot C_{(j,A)}$.
We now distinguish two cases.\\
\textbf{Case 1.} There exists $i$ among $i^1, i^2, \ldots, i'$ such that $\overline{y}_i = 1$.
Then facility $i$ will be deterministically opened by the algorithm.
Note that since $x^*_{A,ij} >0$, we have $i \notin A$ (i.e., facility $i$ 
is not closed by the adversary
in scenario $A$); hence, we can connect $j$ to $i$ in scenario $A$ in our constructed integral solution. 
It remains to observe that $c_{ij} \leq \frac{(k+4)(k+1)}{3k} \cdot C_{(j,A)} \leq (k+5+4/k) \cdot C_{(j,A)}$ 
is a distance that we can accept.

\smallskip \noindent
\textbf{Case 2.} There is no $i$ among $i^1, i^2, \ldots, i'$ such that $\overline{y}_i = 1$.
Then we have $\overline{x}_{A,i^1j} + \overline{x}_{A,i^2j} + \cdot + \overline{x}_{A,i'j} \geq k+1$,
which is the fractional connection to at least $k+1$ facilities, each of them
within the distance of $\frac{(k+4)(k+1)}{3k} \cdot C_{(j,A)}$. 
With a randomized rounding technique described below, they will be turned into
$k+1$ facilities opened within the distance of $3 \cdot \frac{(k+4)(k+1)}{3k} \cdot C_{(j,A)} = 
(k+5+4/k) \cdot C_{(j,A)}$. Since at most $k$ of these facilities will be closed by the adversary,
there remains an open facility for client $j$ in scenario $A$ at distance at most $(k+5+4/k) \cdot C_{(j,A)}$.

It remains to argue that we can turn $k+1$ fractional connections to facilities at distance at most $d$
into $k+1$ integral connections to facilities at distance at most $3d$. This can be seen
as a situation typical for LP-rounding algorithms for the standard fault-tolerant facility location problem.
Indeed, exactly this property is associated with the rounding scheme in~\cite{ipco10} (Sections 3 and 4). 
It is obtained by carefully constructing
a laminar family of subsets of facilities and performing a dependent rounding procedure guided
by the subsets. It can also be thought of as an application of the pipage-rounding technique~\cite{AS}. 

\subsection{Better bound in the oblivious setting}
\label{sec:rftfl-oblivious}
Let us now consider the \emph{oblivious} setting where the $k$ facilities to 
close/fail are chosen
without the knowledge of our opening of facilities. 
In this setting we give a bound on the expected connection cost, where the expectation is 
over the random choices of the algorithm. More precisely, we will argue that the expected
connection cost of client $j$ in scenario $A$ is bounded with respect to the fractional connection
cost of $j$ in scenario $A$. 

The difference with the previous setting is that now we can use the argument that after scaling the facility-opening variables
by a constant $\gamma$, for a client $j$ in scenario $A$, with probability at least $1-e^{-\gamma}$,
at least one facility from those fractionally serving $(j,A)$ will be opened. Moreover, we can bound
the expected distance to such facility by the fractional connection cost of $(j,A)$.
This allows us to use those facilities that get opened with certainty 
(as described in Section~\ref{sec:rftfl-general}) 
only with a certain small probability. 
In such a situation, it is beneficial to scale the 
facility-opening variables by a little less.

The algorithm is like in Section~\ref{sec:rftfl-general}),
only the scaling parameter $\gamma$ is smaller (say $\gamma=1.5+k$),
and the analysis is different. We argue that 
for every client $j$ in each single 
scenario (choice of the $k$ facilities to close) $A$ the expected
connection cost is bounded. As before we distinguish two cases.\\
{\bf Case 1.} (There exists $i$ such that $\overline{x}_{A,ij} = 1$)
If there is such facility at distance at most $(1.5+k) C_{j,A}$ we just connect to it,
otherwise, the average connection cost in $\overline{x}$ is only smaller then the average connection cost in $x^*$,
and we may use a version of the Lemma~\ref{exp_distance_lemma} to argue
that the expected connection cost to the closest of the facilities 
randomly opened by the algorithm is at most $C_{j,A}$.

\smallskip \noindent
{\bf Case 2.} (There is no such facility, 
and therefore there is no $i$ among $i^1, i^2, \ldots, i'$ such that $\overline{y}_i = 1$)
Then we have $\overline{x}_{A,i^1j} + \overline{x}_{A,i^2j} + \cdot + \overline{x}_{A,i'j} \geq k+1$,
which is the fractional connection to at least $k+1$ facilities, each of them
within the distance of $(3+2k) \cdot C_{(j,A)}$. Just like in Section~\ref{sec:rftfl-general}),
we argue that as a result of dependent rounding we obtain $k+1$ facilities 
deterministically opened within the distance $3(3+2k)\cdot C_{(j,A)}$.
And now we propose a suboptimal assignment procedure to bound the cost of the optimal one.
In the suboptimal assignment, client first looks at facilities fractionally serving
him. If one of them is opened then connect to the closest one, which would 
incur an expected cost of $C_{(j,A)}$, and if non is open, then take a facility
deterministically opened within distance $3(3+2k)\cdot C_{(j,A)}$.
Like for the other results in this paper, we then argue that the expected connection cost is 
at most $(1-e^{-(1.5+k)})\cdot C_{(j,A)} + e^{-(1.5+k)}\cdot3(3+2k)\cdot C_{(j,A)}$,
which is less then $(1.5 + k) \cdot C_{(j,A)}$ for $k\geq2$.

\subsection{Integrality gap example}
Let us now show that the program (\ref{eqn:rf0})-(\ref{eqn:rf3})
has integrality gap at least $k+1-\epsilon$.
Consider the following instance. There is a single client
and $n$ identical facilities. All the facility-opening costs are 1,
and all the connection costs are 0.
The optimal fractional solution opens each facility 
to the extent of $\frac{1}{n-k}$, 
incurring cost $\frac{n}{n-k} \stackrel{n \to \infty}{\rightarrow} 1$.
Any integral solution, however, needs to open at least $k+1$ facilities
and therefore has cost at least $k+1$.
Therefore, for any $\alpha < k+1$ there exists an instance of the $k$-robust fault tolerant problem, for which the integrality gap of the 
program (\ref{eqn:rf0})-(\ref{eqn:rf3}) at least $\alpha$.

\end{document}